\documentclass{article}

\usepackage[english]{babel}
\usepackage{amsmath,amssymb,array}
\usepackage{hyperref}
\usepackage[color=green!30]{todonotes}
\usepackage{xspace}
\usepackage{tikz}
\usepackage{a4wide}
\usepackage{amsthm}
\usepackage{comment}
\usetikzlibrary{positioning}
\usetikzlibrary{arrows,shapes,through,chains,fit,positioning,patterns}

\tikzstyle{vertex}=[circle, fill=white, draw, inner sep=2pt, minimum width=4pt, minimum size=18]

\newtheorem{theorem}{Theorem}
\newtheorem{lemma}[theorem]{Lemma}
\newtheorem{corollary}[theorem]{Corollary}
\newtheorem{definition}[theorem]{Definition}

\renewcommand{\S}{\mathcal{S}}

\newcommand{\N}{\mathbb{N}}
\newcommand{\R}{\mathbb{R}}
\newcommand{\seq}[1]{\langle #1 \rangle}

\renewcommand\leq\leqslant
\renewcommand\geq\geqslant

\date{}

\newenvironment{subproof}[1][\proofname]{%
  \renewcommand{\qed}{\begin{flushright}$\blacksquare$\end{flushright}}%
  \begin{proof}[#1]%
}{%
  \end{proof}%
}

\newcommand{\fullversion}[1]{}
\newcommand{\confversion}[1]{#1}

\usepackage{etoolbox}

\newcounter{Bew1}
\newcounter{Bew2}
\newcounter{Def1}

\newcommand{\appendixproof}[2]{

   \stepcounter{Bew1}
   \label{L\arabic{Bew1}}
   \gappto{\appendixProofText}{ \phantomsection \stepcounter{Bew2}\label{\arabic{Bew2}} \subsection{Proof~\textbf{\arabic{Bew2}} (#1)} #2}
	\vspace{-2pt}
}

\newcommand{\problemopt}[3]{%
\begin{center}
  \begin{tabular}{|l|}%
  \hline
    \begin{minipage}[c]{.97\linewidth}
      \smallskip%
      \par\noindent%
      \textsc{#1}
      \medskip%
      \par\noindent%
      $\bullet$
      \textbf{\textsf{Input}}: #2%
      \par\noindent%
      $\bullet$
      \textbf{\textsf{Output}}:
      #3%
      \smallskip%
      \par\noindent%
    \end{minipage}
  \\\hline
  \end{tabular}%
\end{center}
}%

\begin{document}

\title{Finding Disjoint Paths on Edge-Colored Graphs:  More Tractability Results\thanks{An extended abstract of this work appears in~\cite{Dondi2016}}}

\author{Riccardo Dondi\footnote{\noindent Dipartimento di Scienze umane e sociali,  Universit\`a degli Studi di Bergamo, ITALY \tt{riccardo.dondi@unibg.it}} \and Florian Sikora\footnote{Universit\'{e} Paris-Dauphine, PSL Research University, CNRS, LAMSADE, PARIS, FRANCE,  \tt{florian.sikora@dauphine.fr}}
}

\maketitle

\begin{abstract}
  The problem of finding the maximum number of vertex-disjoint uni-color
  paths in an edge-colored graph (called \textsc{MaxCDP}) has been recently introduced in
  literature, motivated by applications in social network analysis.
  In this paper we investigate how the complexity of the problem depends
  on graph parameters (namely the number of vertices to remove to make the graph a collection of disjoint paths and the size of the vertex cover of the graph), which makes sense since graphs in social networks are not random and have structure. 
The problem was known to be hard to approximate in polynomial time and not fixed-parameter tractable (FPT) for the natural parameter. 
Here, we show that it is still hard to approximate, even in FPT-time.
Finally, we introduce a new variant of the problem, called \textsc{MaxCDDP}, whose goal is to find the maximum number of vertex-disjoint and color-disjoint uni-color paths. 
We extend some of the results of \textsc{MaxCDP} to this new variant, and we prove that unlike \textsc{MaxCDP}, \textsc{MaxCDDP} is already hard on graphs at distance two from disjoint paths.

\end{abstract}

\section{Introduction}

The analysis of social networks and social media has introduced several
interesting problems from an algorithmic point of view.
Social networks are usually represented as graphs, 
where vertices represent the elements of the network, and edges represent a binary relation
between the represented elements. 
There are several relevant properties that can be considered when analyzing a social networks,
e.g. degree centrality, closeness centrality, eigenvector centrality.
An interesting property is the vertex connectivity of two given vertices, as it is
related to the information flow inside a network and 
vertex connectivity is considered as a measure of the information flow.
Furthermore, group cohesiveness and centrality are two other
fundamental structural properties that are related to vertex connectivity~\cite{HannemanRiddle,WassermanFaust}.
A well-known result in graph theory, Menger's theorem,
shows that vertex connectivity is equivalent to the maximum number of disjoint
paths between two given vertices.

While social networks analyses usually focus on a single type of relation, 
the availability of several social networks lead to the problem of  
integrating the information of several networks into a single network.
A model for studying vertex connectivity in a multi-relational social network
has been introduced by Wu~\cite{Wu12MaxCDP}, 
where different kinds of relations are considered. In the proposed model,
colors are associated with edges of the graph to
distinguish different kinds of relations. In order to study vertex connectivity in such an edge-colored graph, 
Wu~\cite{Wu12MaxCDP} introduced a natural combinatorial problem,
called  \emph{Maximum Colored Disjoint Paths} (\textsc{MaxCDP}), 
that asks for the maximum number of vertex-disjoint paths consisting of edges of the same colors 
(also called uni-color paths) in the input graph between two terminals. 
A related combinatorial problem, called Multicolor Path problem, have been considered in~\cite{Santos2016}, 
with application in WDM optical networks.
Given an edge-colored directed graph, where each edge is associated with a positive weight and a set of colors,
Multicolor Path problem asks for one shortest path 
between a source and a target in the graph, such that edges in the path share a set of at least $k$ colors.
The Multicolor Path problem have also been extended to ask for a set of $p \geq 1$ paths~\cite{Santos2016}. 

The complexity of \textsc{MaxCDP} has been investigated in~\cite{Wu12MaxCDP,DBLP:journals/algorithms/BonizzoniDP13,Gourves2012}.
\textsc{MaxCDP} is polynomial time solvable when the input graph contains exactly one color as it is 
can be reduced to the maximum flow problem~\cite{Wu12MaxCDP}.
When the edges of the input graph are associated with at least two colors, the 
problem is NP-hard~\cite{Wu12MaxCDP}. Moreover, 
the problem is even NP-hard when the number of paths is fixed to 2 (and 
the graph has degree bounded by 4), thus not in the class XP for the parameter number of paths~\cite{Gourves2012}.

The approximation complexity of the problem has also been investigated.
On general instance, \textsc{MaxCDP} is not approximable within factor $O(n^{d})$, where $n$ is the number of vertices
of the input graph, for any constant $0<d<1$~\cite{DBLP:journals/algorithms/BonizzoniDP13}. 
Furthermore, \textsc{MaxCDP} is approximable within factor $q$, where
$q$ is the number of colors of the edges of the input graph, but
not approximable within factor $2 - \epsilon$, for any $\epsilon > 0$, when $q$ is a fixed constant~\cite{Wu12MaxCDP}.

Since many real networks, and in particular many social networks, have a bounded diameter, Wu introduced a variant of the \textsc{MaxCDP} problem where the length of the paths in the solution are (upper) bounded by an integer $l \geq 1$~\cite{Wu12MaxCDP}.
When $l \geq 4$ \textsc{MaxCDP} is NP-hard, while it is polynomial time solvable for $l < 4$~\cite{Wu12MaxCDP}.
The bounded length variant of \textsc{MaxCDP} is fixed-parameter tractable 
for the combined parameter number of paths in the solution and $l$~\cite{DBLP:journals/algorithms/BonizzoniDP13}. 
Moreover, this variant does not admit a polynomial kernel unless $NP \subseteq coNP/Poly$,
as it follows from the results in~\cite{DBLP:journals/disopt/GolovachT11}.
Wu~\cite{Wu12MaxCDP} considers also the approximation complexity of the bounded length variant of the problem,
showing that is is approximable within factor $(l-1)/2 + \epsilon$~\cite{Wu12MaxCDP}. 

In this paper, we further investigate the complexity of \textsc{MaxCDP} and of a related problem that we introduce, 
called \textsc{MaxCDDP}. Given an edge-colored graph, \textsc{MaxCDDP}
asks for the maximum number of vertex-disjoint and color-disjoint uni-color paths, where
two uni-color paths are color-disjoint if they having different colors.
We introduce color disjointness of paths as it can be useful how different relations in a network connects two vertices. 
In this case, we are not interested to have more paths of a single color, but rather to compute the maximum number 
of color-disjoint paths between two vertices.

In the spirit of a multivariate complexity analysis~\cite{DBLP:conf/mfcs/KomusiewiczN12,Fellows2013},
we study how the complexity of \textsc{MaxCDP} and \textsc{MaxCDDP} depends  on several parameters. 
It has already been studied how the complexity of \textsc{MaxCDP} depends 
on different parameters (number of colors of each edge, degree of the graph, maximum length of a path).
We believe that it is interesting to take into account the structure of the input graph when studying 
the complexity of these two problems for two main reasons. 
First, real-life networks often exhibit very specific structural property, like ``small-world phenomenon''
or small degree of separation, 
and thus information on the structure of the corresponding graphs can be derived.
Moreover, studying how the complexity of \textsc{MaxCDP} and \textsc{MaxCDDP} depends 
on different parameters is also of theoretical interest, 
since it helps to better understand the complexity of the two problems.

First, we investigate how the complexity of the two problems depends on two graph parameters: 
distance from disjoint paths (that is the number of vertices one has to remove to make the graph a collection of disjoint paths), and the size of vertex cover of the graph. 
In Section~\ref{sec:MAXCDDP_dic} we show that on graphs at distance bounded by a constant from disjoint paths 
\textsc{MaxCDP} admits a polynomial-time algorithm, whereas \textsc{MaxCDDP} is NP-hard even if the distance to 
disjoint path has distance two from  disjoint paths. 
This implies the hardness of \textsc{MaxCDDP} even when the treewidth of the input graph is bounded by a constant.
Moreover, we show that \textsc{MaxCDDP} admits a polynomial-time algorithm when the input graph, after the removal of the target vertex, is a tree. 
In Section~\ref{sec:ParamVC} we show that 
\fullversion{\textsc{MaxCDP} is fixed-parameter tractable when parameterized by the size of the vertex cover of the input graph. Moreover, we show that \textsc{MaxCDDP} admits a parameterized $\frac{1}{2}$-approximation algorithm, 
when parameterized by the size of the vertex cover of the input graph.}
\confversion{\textsc{MaxCDP} is fixed-parameter tractable when parameterized by the size of the vertex cover of the input graph.}
In Section~\ref{sec:MaxCDDP-FPT} we consider the parameterized complexity of the bounded length version of \textsc{MaxCDDP}, for the combined parameter number of vertex and color-disjoint paths of a solution and 
maximum length of a path, and we extend the FPT algorithm for \textsc{MaxCDP} to \textsc{MaxCDDP}.
Finally, we consider the FPT-approximability and 
we show in Section~\ref{sec:FPTInapprox} that both problems are not $\rho$-approximable in FPT time, for any function $\rho$.

\section{Definitions}
\label{sec:Def}

In this section we present some definitions that will be useful in the rest of the paper, 
as well as the formal definition of the two combinatorial problems we are interested in.
First, notice that in this paper, we will consider undirected graphs. 
Given a graph $G=(V,E)$ and a vertex $v \in V$, we denote by $N(v)$ the neighborhood of $v \in V$,
that is $N(v)= \{ u: \{ u,v \} \in E \}$.

Consider a set of colors $C=\{c_1,\dots, c_q  \}$, where $|C|=q$. 
A $C$-edge-colored graph (or simply an edge-colored graph when the set of colors is clear from the context) 
is defined as $G=(V,E,f_C)$, where $V$ denotes the set of vertices of $G$
and $E$ denotes the set of edges, 
and $f_C: E \rightarrow 2^{C}$ is a function, called coloring, that associates
a set of colors in $C=\{c_1,\dots, c_q  \}$ with of each edge in $E$. 
In the remaining part of the paper, we denote by $n$ the size of $V$ and by $m$ the size of $E$.

A path $\pi$ in a $C$-edge-colored graph $G$ is said to be colored by $c_j \in C$ if 
all the edges of $\pi$ are colored by $c_j$. 
A path $\pi$ in $G$ is called a \emph{uni-color} path if there is  a color
$c_j \in C$ such that all the edges of $\pi$ are colored by $c_j$. 
Given two vertices $x,y \in V$, an $xy$-path is a path between vertices
$x$ and $y$.
In the remaining part of the paper, we will mainly consider a $C$-edge-colored graph $G$,
with two distinct vertices $s$ (the source) and $t$ (the target), and we will consider uni-color paths
between $s$ and $t$, that is uni-color $st$-paths.

Two paths $\pi'$ and $\pi''$ are \emph{internally disjoint} (or, simply,
\emph{disjoint}) if they do not share any internal vertex, while a set $P$
of paths is internally disjoint if the paths in $P$ are pairwise internally
disjoint.
Two uni-color paths $\pi'$ and $\pi''$ are \emph{color disjoint} if they 
are disjoint and they have different colors.

Next, we introduce the formal definitions of the problems we deal with.


\problemopt{\textsc{Max Colored Disjoint Path (MaxCDP)}}{a set $C$ of colors, a $C$-edge-colored graph $G=(V,E,f_C)$ and two vertices $s,t \in V$.}{the maximum number of disjoint uni-color $st$-paths.}


\problemopt{\textsc{Max Colored Doubly Disjoint Path (MaxCDDP)}}{a set $C$ of colors, a $C$-edge-colored graph $G=(V,E,f_C)$, and two vertices $s,t \in V$.}{the maximum number of color disjoint uni-color $st$-paths.}

We will consider a variant of the two problems  where the length of the paths in the solution is (upper) bounded by an integer $l \geq 1$, that is we are interested only in paths bounded by $l$. 
These variants will be denoted by  $l$-\textsc{MaxCDP} and $l$-\textsc{MaxCDDP}.

We have introduced the optimization definition of \textsc{MaxCDP} and \textsc{MaxCDDP}.
In the parameterized definitions of these problems, 
we are also given an integer $k>0$ and we look whether there exists at least $k$ (color) disjoint uni-color $st$-paths.

Moreover, we will consider how the complexity of the problem is influenced 
by the structure of the graph induced by the set $V \setminus \{t\}$; we will denote such graph as $G^{-t}$.

A parameterized problem $(I,k)$ is said \textit{fixed-parameter tractable} (or in the class FPT) with respect to a parameter $k$ if it can be solved in $f(k)\cdot|I|^c$ time (in \textit{fpt-time}), where $f$ is any computable function and $c$ is a constant.
The class XP contains problems solvable in time $|I|^{f(k)}$, where $f$ is an unrestricted function.
We defer the reader to the recent monographs of Downey and Fellows or Cygan et al. for additional information around parameterized complexity~\cite{Downey2013,Cygan15}. 

The natural notion of \textit{parameterized approximation} was introduced quite recently (see the survey of Marx for an overview~\cite{marx-approx}).
Informally, it aims at giving more time than polynomiality to achieve better approximation ratio.
We give the definition of fpt cost $\rho$-approximation algorithm, as in Section~\ref{sec:FPTInapprox} we will rule out the existence of such an algorithm for \textsc{MaxCDP} and \textsc{MaxCDDP}.
This is a weaker notion than fpt-approximation, but notice that 
we will prove negative result (which will thus be stronger).

An NP-optimization problem $Q$ is a tuple $({\cal I}, Sol, val, goal)$, where ${\cal I}$ is the set of instances, $Sol(I)$ is the set of feasible solutions for instance $I$, $val(I,S)$ is the value of a feasible solution $S$ of $I$, and $goal$ is either max or min.

\begin{definition}[fpt cost $\rho$-approximation algorithm, Chen et al.~\cite{DBLP:conf/iwpec/ChenGG06}]
Let $Q$ be an optimization problem  and $\rho\colon \N \rightarrow \R$ be a function such that $\rho(k) \geq 1$ for every $k\geq 1$ and  $k \cdot \rho(k)$  is nondecreasing (when $goal$ = min)
or 
$\frac{k}{\rho(k)}$
 is unbounded and nondecreasing  (when $goal$ = max).
 A decision algorithm ${\cal A}$ is an \emph{fpt cost $\rho$-approximation algorithm} 
 for $Q$ (when $\rho$ satisfies the previous conditions) if for every instance $I$ of $Q$ and integer $k$, with $Sol(I)\neq \emptyset$, its output satisfies the following conditions:

\begin{enumerate}
	\item If $opt(I) > k$ (when $goal$ = min) or $opt(I) < k$ (when $goal$ = max), then ${\cal A}$ rejects $(I,k)$.
	\item If $k\geq opt(I) \cdot \rho(opt(I))$ (when $goal$ = min) or 
	$k\leq \frac{opt(I)}{\rho(opt(I))}$
	 (when $goal$ = max), then ${\cal A}$ accepts $(I,k)$.
\end{enumerate}

 Moreover the running time of ${\cal A}$ on input $(I,k)$ is $f(k) \cdot |I|^{O(1)}$.
If such  a decision algorithm ${\cal A}$ exists then $Q$ is called fpt cost $\rho$-approximable.
\end{definition}

\section{\textsc{MaxCDDP} and \textsc{MaxCDP} on Graphs at Bounded Distance from Disjoint Paths}
\label{sec:MAXCDDP_dic}

In this section we consider how the complexity of \textsc{MaxCDP} and \textsc{MaxCDDP} on graphs
depends on some strong structural properties: distance bounded by a constant from disjoint paths
and graph $G^{-t}$ being a tree.
The distance to disjoint paths is the minimum number of vertices to remove to make the graph a set of disjoint paths.
The \textsc{MaxCDP} problem on a set of disjoint paths is trivially polynomial time solvable. 
Thus it is natural to consider the complexity of 
\textsc{MaxCDP} and \textsc{MaxCDDP} on graphs that are at small distance from this parameter.

First, we consider the \textsc{MaxCDDP} problem and 
we show that \textsc{MaxCDDP} is in P when the graph induced by the set $V \setminus \{t\}$ is a tree, 
while it is NP-hard for graphs at distance two from disjoint paths.



{
\subsection{\textsc{MaxCDDP} on Trees}

We now show that \textsc{MaxCDDP} is polynomial time solvable when the input graph $G^{-t}$ is a tree. 
Notice that \textsc{MaxCDP} when $G^{-t}$ is a tree is trivially solvable in polynomial time.

\begin{theorem}
Given an input graph $G$, \textsc{MaxCDDP} is in $P$ when $G^{-t}$ is a tree.
\end{theorem}
\begin{proof}

Consider the tree $G^{-t}$ rooted at $s$ and let $N(s)= \{v_1, \dots,  v_p\}$. 
Now, define a bipartite graph $G_B= (V_B, E_B)$ consisting of vertices:
\[
V_B= \{v_i : \{s,y_i\} \in E  \} \cup \{ v'_i : c_i \in C  \}
\]
The set $E_B$ of edges is defined as follows:
\[
E_B= \{ \{ v_i, v'_j\} : \text{there exists a uni-color path from $s$ to $t$ colored $c_j$  that passes through $y_i$}  \}
\]

We show that there exists a matching $M$ in $G_B$ of size $h$ if and only if there exists a solution of \textsc{MaxCDDP} in $G^{-t}$ consisting of $h$ paths.

Consider a matching $M$ in $G_B$ consisting of $h$ edges.
Construct a set $P$ of $h$ paths in $G$ as follows: if $\{ v_i, v'_j\} \in M$ then add in $P$ the path colored $c_j$ from $s$ to $t$ that passes through $y_i$.
First, notice that since the $G^{-t}$ is a tree, the paths associated with edges of the matching are disjoint,
as by construction there exists at most one path that passes through a node $y_i \in N(s)$.
Next, we show that $P$ is a set of color disjoint paths.
By construction each edge of $M$ can be incident in at most one 
vertex $v'_j$. Then there exists at most one path in $P$ having color 
$c_j$, hence $P$ is a set of color disjoint paths.


Consider a set $P$ of $h$ uni-color and color disjoint paths that is a 
solution of \textsc{MaxCDDP}. 
We compute a matching $M$ of $G_B$ as follows: 
given a path $\pi$ of $P$ colored $c_j$  that passes through $y_i$,
add the edge $\{ v_i, v'_j\}$ to $M$.
%
Next, we show that $M$ is a matching of $M$ (obviously $M$ consists of $h$ edges).
Since the paths in $P$ are color disjoint, the following properties hold: (1)
for each $v_i \in V_B$, there exists at most one edge of $M$ incident in $v_i$,
since $G^{-t}$ is a tree, hence there exists at most one path that
passes through $y_i$; (2) for each $v'_j \in V_B$, 
there exists at most one edge of $M$ incident in $v'_j$,
since the paths in $P$ are color disjoint, hence there exists at most
one path in $P$ colored $c_j$.

Therefore, it follows that an optimal solution $P$ of \textsc{MaxCDDP} can be computed by first finding a maximum matching $M$ of $G_B$ and then assigning to $P$ those paths whose corresponding edges belong to $M$.
\qed
\end{proof}

}

\subsection{Complexity of \textsc{MaxCDDP} on Graphs at Distance Two from Disjoint Paths} 

In this section, we show that if the input graph $G$ has distance two from a set of disjoint paths, 
then \textsc{MaxCDDP} is NP-hard. 

We give a reduction from Maximum Independent Set on Cubic graphs (\textsc{MaxISC}).
We recall that a graph is cubic when each of its vertices has degree $3$. 
Moreover, we give the definition of the \textsc{MaxISC} problem:

\problemopt{\textsc{Maximum Independent Set on Cubic graphs (MaxISC)}}{a cubic graph $G_I=(V_I,E_I)$.}{a subset $V'_I \subseteq V_I$ of maximum cardinality,
such that for each $v_x, v_y  \in V'_I$ it holds $\{ v_x, v_y \} \notin E$}


In this section we give a polynomial reduction from \textsc{MaxISC} to \textsc{MaxCDDP}. 
We build a graph $G=(V,E,f_C)$ (input of \textsc{MaxCDDP}) 
from $G_I=(V_I,E_I)$, by defining a gadget $GV_i$ for each vertex $v_i \in V_I$,
and connecting the gadget to vertices $s$ and $t$.

Given $v_i \in V_I$, define a gadget $GV_i$ consisting of a set $V_i$ of $4$ vertices 
(see Figure~\ref{fig:RedBoundTreewidth}):
\[
V_i =  \{ v'_i, v'_{i,j}: v_i \in V_I, 1 \leq j \leq 3 \} 
\]


Moreover, define the set $C$ of colors as follows:
\[
C = \{ c_i : v_i \in V_I \} \cup \{ c_{i,j}: \{ v_i,v_j \} \in E_I \}
\]

We assume that, given a vertex $v_i$, the vertices adjacent to $v_i$ (that is the vertices in $N(v_i)$) are ordered,
i.e. if $v_j,v_h,v_z \in N(v_i)$ with $1 \leq j \leq h \leq z$, then $v_j$ is the first vertex adjacent
to $v_i$, $v_h$ is the second and $v_z$ is the third.

We define the edges of $G$ and their colors by means of the following colored paths:
\begin{itemize}
\item a path colored $c_i$ that consists of $s$, $v'_i$, $v'_{i,1}$, $v'_{i,2}$, $v'_{i,3}$, $t$, with $1 \leq i \leq |V_I|$
\item if, according to the ordering, $v_j$ is the $p$-th vertex incident on $v_i$, $1 \leq p \leq 3$, then there exists a path
colored $c_{i,j}$ that passes through $s$, $v'_{i,p}$, $t$
\end{itemize}

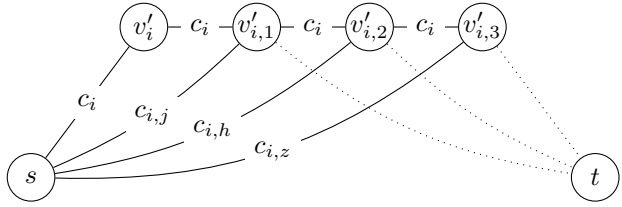
\begin{figure}[ht!]
\centering

		\begin{tikzpicture} 

\node[vertex,inner sep=-5,minimum size=18pt] (s) at (-0.5,-1)	{$s$};

\node[vertex,inner sep=-5,minimum size=18pt] (v0) at (1,1)	{$v'_i$};

\node[vertex,inner sep=-5] (v1) at (2.5,1)	{$v'_{i,1}$};

\node[vertex,inner sep=-5] (v2) at (4,1)	{$v'_{i,2}$};

\node[vertex,inner sep=-5] (v3) at (5.5,1)	{$v'_{i,3}$};

\node[vertex,inner sep=-5,minimum size=18pt] (t) at (7,-1)	{$t$};

\draw (s) edge[] node[midway,fill=white] {$c_i$} (v0);
\draw (v0) edge[] node[midway,fill=white] {$c_i$} (v1);
\draw (v1) edge[] node[midway,fill=white] {$c_i$} (v2);
\draw (v2) edge[] node[midway,fill=white] {$c_i$} (v3);
\draw (s) edge[bend right=10] node[midway,fill=white] {$c_{i,j}$} (v1);
\draw (s) edge[bend right=15] node[midway,fill=white] {$c_{i,h}$} (v2);
\draw (s) edge[bend right=20] node[midway,fill=white] {$c_{i,z}$} (v3);
\draw (v3) edge[dotted] (t);
\draw (v2) edge[dotted, bend right=10] (t);
\draw (v1) edge[dotted, bend right=15] (t);

\end{tikzpicture}
\caption{Gadget $GV_i$ associated with vertex $v_i \in V_I$. Vertices $v_j$, $v_h$, $v_z$ are three vertices of $V_I$, 
with $N(v_i)=\{ v_j, v_h, v_z \}$ and $j < h < z$. $v_j$ is the first vertex adjacent to $v_i$ in $G_I$ and thus there 
exists a path in $G$ colored by $c_{i,j}$ that passes through $s$, $v'_{i,1}$, $t$;
$v_h$ is the second vertex adjacent to $v_i$ in $G_I$ (hence there exists a path in $G$ colored by $c_{i,h}$ that passes through $s$, $v'_{i,2}$, $t$);
$v_z$ is the third vertex adjacent to $v_i$ in $G_I$ (hence there exists a path in $G$ colored by $c_{i,z}$ that passes through $s$, $v'_{i,3}$, $t$).
}
\label{fig:RedBoundTreewidth}

\end{figure}

First, we prove that the graph $G$ has indeed the desired property, 
that is it has distance two from disjoint paths.

\begin{lemma}
\label{Lem:BoundTr1}
Given a cubic graph $G_I$, let $G$ be the corresponding graph input of \textsc{MaxCDDP}.
Then $G$ is at distance two from disjoint paths.
\end{lemma}
{
\begin{proof}
After the removal of $s$ and $t$, the paths left in the resulting graph are the paths colored by $c_i$, with $1 \leq i \leq |V_I|$, that pass through $v'_i$, $v'_{i,1}$, $v'_{i,2}$, $v'_{i,3}$. 
Since these paths are pairwise vertex disjoint, the lemma holds.
%
%
%
%
\qed
\end{proof}
}
%

Next, we prove the main results of the reduction.

\begin{lemma}
\label{Lem:BoundTrRed1}
Let $G_I$ be a cubic graph and $G$ be the corresponding graph input of \textsc{MaxCDDP}.
Given an independent set $V'_I$ of $G_I$, we can compute in polynomial time $|E|+|V'_I|$ 
disjoint uni-color color paths
in $G$.
\end{lemma}
\begin{proof}
Consider an independent set $V'_I \subseteq V_I$ of $G_I$, we define
a set $\mathcal{P}$ of uni-color disjoint paths as follows:

\begin{itemize}

\item for each $v_i \in V'_I$, $\mathcal{P}$ contains  a path $s,v'_i, v'_{i,1}, v'_{i,2}, v'_{i,3}, t$ colored by $c_i$

\item for each $ \{ v_i,v_j \} \in E_I$, there exists at least one of $v_i$, $v_j$ in $V_I \setminus V'_I$,
we pick one of $v_i$, $v_j$ in $V_I \setminus V'_I$, assume w.l.o.g. that $v_i \in V_I \setminus V'_I$ and that 
$v_j$ is the $h$-th vertex, $1 \leq h \leq 3$, adjacent to $v_i$, then 
$\mathcal{P}$ contains the path $s, v'_{i,h}, t$  colored by $c_{i,j}$

\end{itemize}

%

Since $V'_I$ is an independent set, there exists at most one path colored by $c_i$ in $\mathcal{P}$.
Moreover, notice that by construction, for each $ \{ v_i,v_j \} \in E_I$, there exists
exactly one path colored by $c_{i,j}$ in $\mathcal{P}$.
Thus, in order to prove that $\mathcal{P}$ is color disjoint, we have to prove that
the paths in $\mathcal{P}$ are  internally disjoint.
Indeed, since $V'_I$ is an independent set, a path colored by $c_i$ is the only path that passes 
though vertices $s,v'_i, v'_{i,1}, v'_{i,2}, v'_{i,3}, t$.
Furthermore, consider a path $\pi$ colored by $c_{i,j}$. 
Then $\pi$ passes through vertices 
$s$, $v'_{i,h}$, $t$, where $v_i \notin V'_I$, and by construction there is no other path
that passes through $v'_{i,h}$.
\qed
\end{proof}

\begin{lemma}
\label{Lem:BoundTrRed2}
Let $G_I$ be a cubic graph and $G$ be the corresponding graph input of \textsc{MaxCDDP}.
Given $|E|+t$ color disjoint uni-color paths
in $G$, we can compute in polynomial time an independent set of size $t$ for $G_I$.
\end{lemma}
{
\begin{proof}
Consider a solution $\mathcal{P}$ of the instance of \textsc{MaxCDDP} consisting of $|E|+t$ color disjoint  
uni-color paths.
First, we show that we can assume that $\mathcal{P}$ contains, for each color $c_{i,j}$, a path colored 
by $c_{i,j}$.
Assume this is not the case. Then, we can replace a path colored by $c_i$  
with a path $p'$ colored by $c_{i,j}$ that passes through the vertices of gadget 
$VG_i$, without decreasing the number of path in $\mathcal{P}$. 
Indeed, by the property of $G_I$, there exists only two uni-color paths that passes through 
vertex $v'_{i,h}$, with $1 \leq h \leq 3$, where $\{ v_i,v_j \}$ is the $h$-th edge incident
on $v_i$: a path colored by  $c_i$ that passes through $s,v'_i, v'_{i,1}, v'_{i,2}, v'_{i,3}, t$
and a path $s, v'_{i,h}, t$  colored by $c_{i,j}$.
Hence by replacing a path color $c_i$ with $p'$, the set $\mathcal{P}$ 
$|E|+t$ color disjoint  uni-color paths.

Now, starting from $\mathcal{P}$, we can compute an independent set $V'_I$ as follows.
If $\mathcal{P}$ contains a path $s,v'_i, v'_{i,1}, v'_{i,2}, v'_{i,3}, t$ colored by $c_i$,
then $v_i \in V'_I$. 
Notice that $V'_I$ is an independent set, since, if $v_i$, $v_j$, with $\{ v_i,v_j \} \in E$, 
are both in $V'_I$, this implies that there is no path colored by $c_{i,j}$ in $\mathcal{P}$,
contradicting our assumption.
\qed
\end{proof}
}

Hence, we can prove the NP-hardness of \textsc{MaxCDDP} on graphs at distance two from disjoint paths. 

\begin{theorem} 
\label{th:hardnesstw}
\textsc{MaxCDDP} is NP-hard, even if the graph $G$ has distance two from disjoint paths.
\end{theorem}
{
\begin{proof}
 \textsc{MaxISC} is NP-hard~\cite{DBLP:journals/tcs/AlimontiK00}.
Hence Lemma~\ref{Lem:BoundTr1}, Lemma~\ref{Lem:BoundTrRed1} and Lemma~\ref{Lem:BoundTrRed2} imply that \textsc{MaxCDDP} is NP-hard, even if the graph $G$ 
has distance two from disjoint paths.
\qed
\end{proof}
}

The previous result implies that \textsc{MaxCDDP} cannot be solved in $n^{f(d)}$ time unless P=NP (it is not in the class XP), where $d$ is the distance to disjoint paths of $G$, but also for ``stronger'' parameters like pathwidth or treewidth~\cite{DBLP:conf/mfcs/KomusiewiczN12}.

\begin{corollary}
\label{corollary:hardnesspw}
\textsc{MaxCDDP} is NP-hard, even if the input graph $G$ has pathwidth and treewidth bounded by $2$. 
\end{corollary}
%

%

\subsection{A Polynomial-Time Algorithm for \textsc{MaxCDP} on Graphs at Constant Distance from Disjoint Paths}
\label{sec:MaxCDPConstDisjoint}

In this section, we show that, contrary to \textsc{MaxCDDP}, \textsc{MaxCDP} is polynomial-time solvable when the input graph $G$ has distance bounded by a constant $d$ from a set $\mathcal{P}$ of disjoint paths (that is, it is in the class XP for the parameter distance to disjoint paths).

Next, we present the algorithm. 
Notice that we assume that a set $X \subseteq V$ is given, such that after the removal of $X \cup \{s,t\}$ the resulting graph
consists of a set $\mathcal{P}$ of disjoint paths \footnote{Notice that, since $|X| \leq d$, $X$ can be computed in time $O(n^d)$ .}.
We assume that $X$ and $\mathcal{P}$ are defined so that
$s,t \notin X$ and that no path in $\mathcal{P}$ contains $s$ and $t$.
Denote by $V(\mathcal{P})$ the set of vertices that belong to a path of $\mathcal{P}$, 
it holds $V= V(\mathcal{P}) \cup X \cup \{s,t\}$.

Since $G$ has distance $d$, where $d>0$ is constant, from the set of disjoint paths $\mathcal{P}$, it follows that
$|X|\leq d$. 
Let $\mathcal{P'} =  \{p_1, \dots, p_b \}$, with $1 \leq i \leq b \leq d$,  such that
$V(\mathcal{P'}) \subseteq V$
is the set of paths of an optimal solution of \textsc{MaxCDP} such that
$p_i$ contains a non-empty subset of $X$.

The algorithm computes each $p_i$, with $1 \leq i \leq b$, by iterating through subpaths of size at most $d$ in $\mathcal{P}$ and a subset of $X$.
More precisely, $p_i$ is computed as follows. 
Each path $p_i$ contains at most $d+1$ disjoint subpaths that belong to paths in $\mathcal{P}$, that are connected
through a subset of at most $d$ vertices of $X$. 
In time $O(n^{2(d+1)})$, we compute the at most $d+1$ disjoint subpaths $p_x[j_1,j_2]$ of $P_x \in \mathcal{P}$ that belong to $p_i$; in time $O(2^d)$ we compute the subset $X_i \subseteq X$ that belong to each $p_i$.
Let $V_i=V(p_i) \cup X_i$, that is the set of vertices that belong to $p_i$ and to subset $X_i$.
Notice that the subsets $V_i$, with $1 \leq i \leq b$, are computed so that they are pairwise disjoint.

The algorithm computes in polynomial time if there exists a uni-color path from $s$ to $t$ that passes through the vertices $V_i$.
If for each $i$ with $1 \leq i \leq b$ such a path exists, then the algorithm computes the
maximum number of uni-color disjoint paths in the subgraph $G'$ of $G$ induced by 
$V' = V \setminus \bigcup_{i=1}^{b} V_i$. Notice that, since $V' \cap X = \emptyset$, it follows
that, if we remove $s$ and $t$ from $V'$, $G'$ consists of a set of disjoint paths $\{p'_1, \dots, p'_r\}$.
The maximum number of uni-color disjoint paths in the subgraph $G'$ can be computed in polynomial
time, as shown in the following lemma.


\begin{lemma} 
\label{lem:MaxCDPDisPath}
Let $G=(V,E,f_c)$ be an edge colored graph such that $V^* = V \setminus \{ s,t \}$ induces
a set of disjoint paths. Then \textsc{MaxCDP} on $G$ can be solved in polynomial time.
\end{lemma}
{
\begin{proof}
Let $P = \{ p'_1, \dots , p'_r \}$ be the set of disjoint paths induced by $V^*$. 
Since there is no $st$-path in $G$ containing a vertex of $p'_i$ and a vertex of $p'_j$, with $i \neq j$,
we compute a solution of \textsc{MaxCDP} independently on each path $p'_i$.
Let $\mathcal{P}_i$ be the set of uni-color $st$-paths that contains only vertices of $p'_i$.
For each $i$ with $1 \leq i \leq r$, we compute a shortest uni-color $st$-path $p$ that contains only vertices of $p'_i$,
we add it to $\mathcal{P}_i$, and we remove the vertices of $p$ from $p'_i$.
We iterate this procedure, until there exists no $st$-path that contains only vertices of $p'_i$.

We claim that $\mathcal{P}_i$ is a set of uni-color $st$-paths of maximum size.
Consider a shortest path $p$ added to $\mathcal{P}_i$. Let $x$ be the vertex of $p$ adjacent to $s$
and $y$ be the vertex of $p$ adjacent to $t$. Notice that each vertex in $p$, except for $x$ and $y$,
is not connected to $s$ or $t$, otherwise $p$ would not be a shortest path between $s$ and $t$.
Now, assume that there is an optimal solution $\mathcal{Q}$ of \textsc{MaxCDP} that does not contain $p$ and that, moreover,
contains an $st$-path that passes through some vertex of $p$, otherwise we can add $p$ to $\mathcal{Q}$
and $\mathcal{Q}$ would not be optimal. Then by construction,
since $P$ is a set of disjoint paths, $\mathcal{Q}$ must contain a path $p'$ that contains $p$ as a subpath.
But then we can replace $p'$ with $p$ in $\mathcal{Q}$, without decreasing the size of the optimal solution.
\qed
\end{proof}
}

Now, we give the main result of this section.

\begin{theorem} 
\label{teo:AlgoXPDisjointPath}
\textsc{MaxCDP} is in XP when the distance to disjoint paths is bounded by a constant.
\end{theorem}
{
\begin{proof}
Notice that for each $i$ with $1 \leq i \leq b \leq d$, we compute the set $V_i$ in time 
$O(2^d n^{2(d+1)})$; hence the $d$ disjoint sets $V_1,\dots, V_b$ are computed in time $O(2^{d^2} n^{2d(d+1)})$.
Since the existence of a uni-color path that passes through the vertices $V_i$ can be computed in polynomial time and 
since  by Lemma~\ref{lem:MaxCDPDisPath} we compute in polynomial time the maximum number of uni-color disjoint paths in the subgraph $G'$, 
the theorem holds.
\qed
\end{proof}
}

\confversion{
\section{FPT Algorithm Parameterized by Vertex Cover for \textsc{MaxCDP}}
\label{sec:ParamVC}
}

\fullversion{
\section{FPT Algorithms Parameterized by Vertex Cover}
\label{sec:ParamVC}
}

\fullversion{In this section we consider the parameterized complexity of \textsc{MaxCDP} 
and \textsc{MaxCDDP} when parameterized by the size of the vertex cover of the input graph. 
First, we show that \textsc{MaxCDP} is FPT when parameterized by the size of the vertex cover of the input graph.
Then, we show that \textsc{MaxCDDP} can be approximated within factor a $\frac{1}{2}$
by a parameterized algorithm, when parameterized by the size of the vertex cover of the input graph.
We start by describing a simple algorithm for \textsc{MaxCDP}
when parameterized by the size of the vertex cover of the input graph.
}

\confversion{
In this section, we will show that \textsc{MaxCDP} is FPT when parameterized by the size of the vertex cover of the input graph.
}

%

\begin{theorem}
\label{teo:MaxCDP-VC}
\textsc{MaxCDP} is in FPT when parameterized by the size of the vertex cover of the input graph.
\end{theorem}
\begin{proof}
We give a constructive proof, by giving a fixed-parameter tractable algorithm for \textsc{MaxCDP}
when parameterized by the size of the vertex cover of the input graph.
The algorithm first considers uni-color paths of length three $s,v,t$, for some $v \in V$. 
Uni-color path of length three
are greedily added to a solution of \textsc{MaxCDP}. Since any solution of \textsc{MaxCDP} 
contains at most one uni-color path that passes through vertex $v$,
it follows that there exists an optimal solution of \textsc{MaxCDP} that contains path $s,v,t$.
Hence, the fixed-parameter tractable algorithm adds 
such path to our solution $\mathcal{P}$, and it removes vertex $v$ from $G$. Let $G'=(V',E',f'_C)$
be the graph after the removal of vertices.

Let $V_C$ be a vertex cover of the resulting graph $G'=(V',E',f_C)$, $|V_C|=k$ (which can be computed in FPT-time).
Since in $G$ there is no uni-color path of length three connecting $s$ and $t$,
the following property holds. Consider a uni-color path $p$ of $G$, 
then $p$ either consists of vertices in $V_C$ or each vertex of $V' \setminus V_C$ that belongs to $p$ 
is adjacent in $p$ to vertices of $V_C \cup \{ s \} \cup \{ t \}$.
This is true since $V_C$ is a vertex cover (and thus $V' \setminus V_C$ is an independent set in $G$).

A consequence of 
this property  is that each uni-color path in $G'=(V',E',f'_C)$ has length at most $2k$.
Moreover, there can be at most $k$ uni-color paths in a solution of \textsc{MaxCDP} on instance
$G'=(V',E',f'_C)$, since each path must contain a vertex of $V_C$ and $|V_C| \leq k$.
Since both the number of paths and the length of paths are bounded by $k$ and \textsc{MaxCDP} is known to be in FPT w.r.t. the combination of these two parameters~\cite{DBLP:journals/algorithms/BonizzoniDP13}, the claimed result follows.
\qed
\end{proof}


\confversion{
This algorithm does not easily extend to \textsc{MaxCDDP}.
}
The main difference between \textsc{MaxCDDP} and \textsc{MaxCDP}, when considering as parameter the vertex cover of the input graph, is that in the latter we can safely add a uni-color path $s,v,t$ 
of length three  to a solution,  while in the former we are not allowed to do it.
Consider for example the uni-color path $s,v,t$ of length three colored by $c$; 
if this path belongs to a solution of \textsc{MaxCDDP}, 
it prevents any other uni-color path $p'$ that passes through $v$ (colored by some color $c'$), but also any path $p''$ colored by $c$ 
(that does not pass through $v$) to be part of the solution. 
So, by adding the path $s,v,t$ to the solution we are computing, we may get a suboptimal solution, since by removing $p$ and by adding $p'$ and $p''$, we possibly compute a larger set of disjoint color uni-color paths (see Figure~\ref{fig:maxcddp} for an illustration).

\begin{figure}
\centering
\begin{tikzpicture}
\tikzstyle{vertex}=[draw,circle,fill=black!0,minimum size=14pt,inner sep=0pt]
\node[vertex] (s) at (1,0) {$s$};
\node[vertex] (v) at (3,1) {$v$};
\node[vertex] (t) at (5,0) {$t$};

\node[vertex] (u) at (3,0) {$u$};

\draw (s) edge[ultra thick, bend left=15, red] (v);
\draw (v) edge[ultra thick, bend left=15,red] (t);

\draw (s) edge[ultra thick, red] (u);
\draw (u) edge[ultra thick, red] (t);

\draw (s) edge[ultra thick, green!60, bend right=15] (v);
\draw (v) edge[bend left,ultra thick, green!60, bend right=15] (t);

\end{tikzpicture}
\caption{An example where adding first the red (dark gray) path $s,v,t$ avoids the better solution with two paths: $s,v,t$ colored green (light gray) and $s,u,t$ colored red (dark gray).}\label{fig:maxcddp}
\end{figure}
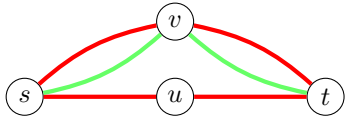

On the other side, we show next that the algorithm described in Theorem~\ref{teo:MaxCDP-VC} returns a solution consisting of at least  half of the paths in an optimal solution of \textsc{MaxCDDP}.
The existence of an exact parameterized algorithm with respect to the size of the vertex cover of the input graph for \textsc{MaxCDDP} remains however open.

\begin{theorem}
\textsc{MaxCDDP} admits an approximation fixed-parameter algorithm parameterized by the size of the vertex cover of the input graph of factor $\frac{1}{2}$.
\end{theorem}
\begin{proof}
As in Theorem~\ref{teo:MaxCDP-VC} we give a constructive proof.
We will show that there exists a fixed-parameter algorithm, parameterized by the size of the vertex cover of the input graph, that returns a set $\mathcal{P}$ of color disjoint paths of the input graph $G=(V,E,f_C)$ such that, 
given an optimal solution $OPT$ of  \textsc{MaxCDDP} on instance $G=(V,E,f_C)$, it holds that $|\mathcal{P}| \geq \frac{1}{2}|OPT|$.

As for the algorithm of Theorem~\ref{teo:MaxCDP-VC}, 
the fixed-parameter algorithm first considers uni-color paths of length three $s,v,t$, for some $v \in V$. 
First, the fixed-parameter algorithm defines a set $APPROX_A$ of color disjoint uni-color paths of length three,
by greedily adding uni-color paths of length three to $APPROX_A$. 
Define the set of colors 
\[
C'=  \{c: \text{ there exists a path in $APPROX_A$ colored by $c$} \}
\]

Now, the algorithm constructs an instance of \textsc{MaxCDDP} as follows.
First, define the set $C_H$ of remaining colors: $C_H=C \setminus C'$.
Now, let $H=(V_H,E_H,f'_C)$ be the subgraph of $G$ obtained from $G=(V,E,f_C)$ as follows: 
\begin{description}

\item[Step 1] Remove the vertices in $V \setminus \{s,t\}$ that induce paths in $APPROX_A$ 
and all the edges having a color in $C'$.

\item[Step 2] Remove the vertices that, after Step 1, do not belong to a uni-color path from $s$ to $t$ colored 
by some $c \in C_H$.
\end{description}

 
Similarly to Theorem~\ref{teo:MaxCDP-VC}, we compute (in FPT-time) a vertex cover $V_C$ of $H$, with $|V_C|=k$.
Our algorithm computes a maximum cardinality set $APPROX_H$ of color disjoint paths in $H=(V_H,E_H,f'_C)$ having a color in $C_H$. 
We will next show that this is possible in FPT time, where the parameter is $k=|V_C|$, the size of the vertex cover of $H$.
Notice that there is no uni-color path $s,v,t$ of length $3$ in $H$ with color in $C_H$, otherwise it would have been added to $APPROX_A$.
Let $u$ be a vertex of $V_H$.
Since there is no uni-color path of length three that connects $s$ and $t$, either $u \in V_C$ or $u$ is adjacent to at least one vertex of $V_C$.
It follows that each uni-color path in $H$ has length at most $2k$.
Moreover, there can be at most $k$ color disjoint uni-color paths in a solution 
of \textsc{MaxCDDP} on instance $H=(V_H,E_H,f'_C)$, since each path must contain a vertex of $V_C$ and $|V_C| = k$.
Since both the number of paths and the length of paths are bounded by $k$ and \textsc{MaxCDDP} is shown to be in FPT with respect to the combination of these two parameters
(see Section~\ref{sec:MaxCDDP-FPT}), 
we can compute in FPT time, where the parameter is the size of the vertex cover of $H$, a maximum cardinality set $APPROX_H$ of color disjoint paths in $H=(V_H,E_H,f'_C)$.

The algorithm returns the solution $APPROX=APPROX_A \cup APPROX_H$. Notice that by construction
$APPROX$ is a set of color disjoint paths of $G$.
Next, we show that $|APPROX|$ is at least $\frac{1}{2}|OPT|$,  where $OPT$ is an optimal solution of \textsc{MaxCDDP}  on instance $G$.

We partition $OPT$ in two subsets $OPT_A$ and $OPT_H$ as follows:
let $OPT_A$ be the set of uni-color paths that either contain a vertex not in $V_H$
or have a color not in $C_H$; let $OPT_H = OPT \setminus OPT_A$.
Clearly it holds $|APPROX_H| \geq |OPT_H|$, as each uni-color path in $OPT_H$  is
a uni-color path of $H$ and $APPROX_H$ is a maximum cardinality set of color disjoint paths in $H$.

Consider now a uni-color path $p$ of $OPT_A$. 
If $p$ is colored by $c \in C'$, then by construction of $C'$ there exists
a uni-color path, colored by $c$, in $APPROX_A$.
Hence, assume that $p$ is colored by $c \notin C'$.
Next, we show that $p$ must contain a vertex in $u \in V \setminus V_H$, 
such that there exists a uni-color path in 
$APPROX_A$ that contains $u$.
Assume that this is not the case. 
Recall that if $p$ contains a vertex $v$ in $V \setminus V_H$, then $v$ is removed in the construction of graph $H$.
Hence, in $G$, there is no path having color in $C_H$ that passes through $v$, otherwise
it would not have been removed from $G$ in the construction of $H$, and  thus $p$ would have been in $OPT_H$.

It follows that, for each path in $APPROX_A$, $OPT_A$ 
contains at most two color disjoint uni-color paths.
%
Indeed, consider a path $s,v,t$ in $OPT_A$, for some $v \in V \setminus V_H$, colored by $c$, 
for some $c \in C'$. Since $OPT_A$ is a set of color disjoint paths, 
$OPT_A$ contains at most one path colored by $c$ and at most one path 
that passes through $v$ (recall that $APPROX_H$ contains only paths of length $3$). 
Thus $|OPT_A| \leq 2 |APPROX_A|$.
Since $|OPT_H| \leq |APPROX_H|$, it follows that $|OPT| \leq 2 |APPROX|$.
\qed
\end{proof}

\section{A Fixed-Parameter Algorithm for $l$-\textsc{MaxCDDP}}
\label{sec:MaxCDDP-FPT}

In this section, we  give a fixed-parameter algorithm for $l$-\textsc{MaxCDDP}, the length-bounded version of \textsc{MaxCDDP}, parameterized by the number $k$ of uni-color color disjoint $st$-paths of a solution.
Notice that \textsc{MaxCDDP} is W[1]-hard when parameterized by $k$, as the reduction that prove the  W[1]-hardness of \textsc{MaxCDP} parameterized by $k$ consists of paths having distinct colors~\cite{DBLP:journals/algorithms/BonizzoniDP13}.

Next, we present a parameterized algorithm based on the \emph{color coding} technique~\cite{Alon:Yuster:Zwick:1995}.
The algorithm is inspired by the one for  \textsc{MaxCDP}~\cite{DBLP:journals/algorithms/BonizzoniDP13}.
However, in this case we must combine two different labelings,
one to label the vertices that belong to a uni-color path,
one to label the color associated with a uni-color path
of \textsc{MaxCDDP}.

First, we introduce the definition of perfect hash function on which
our algorithm is based.
A family $F$ of hash functions from a set $U$ (the vertex set in the
traditional applications of color coding) to the set $\{l_1, \ldots, l_k\}$
of labels is \emph{$k$-perfect} if, for each subset $U'$ of $U$ with
$|U'| = k$, there exists a hash function $f \in F$ such that
$f$ assigns a distinct label to each element of $U'$. Function $f$ is called
a \emph{labelling function}.

Let $f_v \in F_V$ be a labelling function that assigns to each vertex 
 $v \in V \setminus \{ s,t \}$ a label $f_v(v) \in L_v=\{ 1_v, \dots, h_v \}$, where $h_v=|L_v| \leq l k$.

Consider a second labelling function $f_c \in F_C$ that assigns to each color $c \in C$ a label $f_c(v) \in L_c=\{ 1_c, \dots, h_c \}$, where 
$h_c=|L_c| \leq k$.

By the property of perfect hash functions, we assume that 
each vertex that belongs to a solution of \textsc{MaxCDDP} is assigned 
a distinct label by $f_v$ and that each color $d$, such that there exists 
a uni-color path of \textsc{MaxCDDP} colored by $d$, is associated with distinct label
$f_c$.

A simple path $p$ in $G$ is \emph{perfect} for a set $L_v$ of labels assigned to $V$ if and only if for each vertex $v$ of $p$, with $v \notin \{s,t\}$, $f_v(v) \in L_v$, and for each pair of distinct vertices $u$, $v$ of $p$, $f_v(u) \neq f_v(v)$.
A set $\{ p_1, \ldots , p_k \}$ of uni-color paths is \emph{perfect} for the set $L_v$  and $L_c$ of labels if and only if: 
(1) there exists a partition $\{ L_{v,1}, \ldots, L_{v,k} \}$ of $L_v$ such that each $p_i$ is perfect for $L_i$;
(2) each path $p_i$, with $1 \leq i \leq k$, is colored by $c \in C$ associated with a distinct label in $L_c$. 
We combine two dynamic-programming recurrences to compute, given the labelling functions $f_v$ and $f_c$,
whether there exists a set of perfect uni-color paths in $G$.

First, consider the function $S[L'_v, u, \lambda]$, with $L'_v \subseteq L_v$, $u \in V$ and $\lambda \in C$. 
$S[L'_v, u, \lambda]=1$ if and only if there exists a path from $s$ 
to vertex $u\neq t$, such that the path is perfect for $L'_v$ and $p$ is colored by $\lambda$.

We consider a second function $\Pi[L'_v, M,z] $, with $L'_v \subseteq L_v$ and $M \subseteq L_c$, $0 \leq z \leq k$.
$\Pi[L'_v, M,z] =1$ if and only if there exist a set of labels $L'_v \subseteq L_v$ and a set of labels $M \subseteq L_c$, 
such that there exists a set of $z$ uni-color paths perfect for $L'_v$ and $M$.

$S[L'_v, u, \lambda]$ is defined as follows (we recall that $\uplus$ represents the disjoint union operator).
In the base case, when $u=s$, $S[L'_v, u, \lambda] =  1$ if $L'_v= \emptyset$, otherwise (when $L'_v \neq \emptyset$), 
$S[L'_v, u, \lambda] =  0$.

When $u \neq s$:

\[
S[L'_v, u, \lambda] = \max_{w \in N(u)} \big\{ S[L''_v, w, \lambda] \mid L'_v = L''_v \uplus
  \{f_v(u)\} \wedge \{w,u\}  \text{ is colored by $\lambda$} \} 
\]

Next, we give the recurrence, $\Pi[L'_v,M,z]$. 
In the base case, that is when $z=0$, 
then $\Pi[L'_v,M,0]=1$ if $L'_v = \emptyset$ and $M = \emptyset$,
else $\Pi[L'_v,M,0]=0$.
Recall that $l$ is the bound on the length of each path, $L'_v \subseteq L_v$, $M \subseteq C$ and $0 \leq z \leq k$, 
$\Pi[L'_v,M,z]$ is defined as follows:

\begin{equation}
 \Pi[L'_v,M,z] =
  \begin{cases}
    \begin{aligned}
      \!\max \big\{ &\Pi[L''_v,M \setminus \{ f_c(\lambda) \},z-1] \ \wedge\  S[L^*_v, u, \lambda] \mid \\
      & \
      L'_v = L''_v \uplus L^*_v  \wedge\ |L^*| \leq l-1 \ \wedge\ \lambda \in C \wedge f_c(\lambda) \in M   \ \wedge \ \\
      & \
      \text{$\{u, t\}\in E$ is colored by $\lambda$} 
      \big\}
    \end{aligned} & 
  \end{cases}
  \label{eq:fpt:P}
\end{equation}

Next, we prove the correctness of the two recurrences.
\fullversion{We start by proving the correctness of the recurrence for $S[L'_v, u, \lambda]$.}

\begin{lemma} 
\label{lem:FTPlMAxCCDPrecS}
Given a labelling $f_v$ of the vertices of $G$, a color $\lambda \in C$, a vertex $u$ and a set
$L'_v \subseteq L_v$, there exists a simple path $p$ in $G$ from $s$ to $u$ \emph{perfect} for 
$L'_v$ if and only if $S[L'_v, u, \lambda]=1$.
\end{lemma}
{
\begin{proof}
We prove the lemma by induction on the length of the path $p$.
First, we consider the base case, that is $u=s$. Since $s$ is not associated with a label
in $L_v$, it holds $S[L'_v, u, \lambda]=1$ if and only if $L'_v= \emptyset$.

Consider now the general case and assume that there exists a path $p$ in $G$ perfect for $L'_v$ such that $p$ is colored by $\lambda$.
Consider the last vertex $u$ of $p$, and let $w$ be the vertex adjacent to $u$ in $p$. 
By induction hypothesis, it follows that 
$S[L''_v, w, \lambda]=1$, where $L'_v=L''_v \uplus \{f_v(u) \}$.
By the definition of the recurrence, then $S[L'_v, u, \lambda]=1$.

Assume that $S[L'_v, u, \lambda]=1$. 
By the definition of the recurrence it holds that $S[L''_v, w, \lambda]=1$, where $L'_v=L''_v \uplus \{f_v(u) \}$ 
and there is an edge  $\{ u,w \} \in E$ colored by $\lambda$. 
By induction hypothesis, since $S[L''_v, w, \lambda]=1$, there exists a path $p'$ from $s$ to $w$ perfect for $L''_v$
such that $p'$ is colored by $\lambda$. 
Since  $\{ u,w \} \in E$ is colored by $\lambda$, 
it follows that there exists a simple path $p$ in $G$ from $s$ to $u$ perfect for $L'_v$.
\qed
\end{proof}
}

\fullversion{Now, we prove the correctness of the recurrence for $\Pi[L'_v,M,z]$.}

\begin{lemma} 
\label{lem:FTPlMAxCCDPrecPi}
Given a labelling $f_v$ of the vertices of $G$ and a labelling $f_c$ of the set $C$ of colors, 
a set $L'_v \subseteq L_v$, a set $M \subseteq L_c$, and integer $z$ with $0 \leq z \leq k$, there exists
a set $\{ p_1, \ldots , p_z \}$ of uni-color paths which
is perfect for $L'_v$ and $M$ if and only if  $\Pi[L'_v,M,z]=1$.
\end{lemma}
{
\begin{proof}
We prove the lemma by induction on the number of uni-color paths.
First, we consider the base case, that is $z=0$. 
Then there is no uni-color path perfect for $L'_v = \emptyset$ and $M= \emptyset$ if and only if $\Pi[\emptyset,\emptyset,0]=1$.

Consider now that there exist $z$ disjoint color uni-color paths. 
Consider one of such paths, denoted by $p$, which is colored by $\lambda$ and whose vertices
are associated with set of labels $L^*_v$ and such that 
the vertex of $p$ adjacent to $t$ is $u$, hence $\{ u,t \} \in E$ is colored by $\lambda$. Then, 
by Lemma~\ref{lem:FTPlMAxCCDPrecS} $S[L^*_v, u, \lambda]=1$.
Moreover, by induction hypothesis it holds $\Pi[L''_v,M \setminus \{ f_c(\lambda) \},z-1]=1$, 
where $L'_v = L^*_v \uplus L''_v$ and $f_c(\lambda) \in M$.
Hence, by the definition of the recurrence for $\Pi$, it holds $\Pi[L'_v,M,z]=1$.

Consider the case that $\Pi[L'_v,M,z]=1$. By the definition of function $\Pi$,
it follows that there exists a color $\lambda \in C$, with $f_c(\lambda) \in M$, 
and a set of labels $L^*_v \subseteq L_v$, such that
$\Pi[L''_v,M \setminus \{ f_c(\lambda) \},z-1]=1 $, where 
$L'_v = L''_v \uplus L^*_v $, and $S[L^*_v, u, \lambda]=1$.
By induction hypothesis, since $\Pi[L''_v,M \setminus \{ f_c(\lambda) \},z-1]=1 $,
there exists a set $P'$ of $z-1$ paths perfect for the sets $L''_V$ and $M \setminus \{ f_c(\lambda) \}$.
By Lemma~\ref{lem:FTPlMAxCCDPrecS} there exists a path $p'$ from $s$ to $u$  
colorful for $L''_v$ and that has color $\lambda$. Moreover, since $\Pi[L'_v,M,z]=1$, $\{ u,t \} \in E$ 
is colored by $\lambda$.
By the property of labelling $f_c$, no path of $P'$ has label $f_c(\lambda)$, hence
$P' \cup p$ is perfect for $L'_v$ and $M$.
\qed
\end{proof}
}

We can now state the main result.

\begin{theorem} 
$l$-\textsc{MaxCDDP} can be solved in time $2^{O(l k)}poly(n)$. 
\end{theorem}\label{thm:fptcddp}
{
\begin{proof}
An optimal solution of $l$-\textsc{MaxCDDP} consisting of $k$ color disjoint paths exists if and only if $\Pi[L_v,M,k]=1$.
The correctness of the recurrence to compute $\Pi$ follows from Lemma~\ref{lem:FTPlMAxCCDPrecPi}. 
Now, we discuss the time complexity to compute $\Pi[L'_v,M,z]$ and $S[L'_v, u, \lambda]$.
First, consider $S[L'_v, u, \lambda]$.  It consists of $2^{l k} n q$ entries and each entry can be computed in time 
$O(n)$, as we consider each vertex $w \in N(u)$.

Now, consider $\Pi[L'_v,M,z]$. It consists of $2^{k(l+1)} k$ entries. In order to compute $\Pi[L'_v,M,z]$,
at most $2^{kl} k$ entries must be considered, since  $\Pi[L^*_v,M\setminus \{f_c(v)\},z-1]$ is considered,
where we have $2^{kl}$ subsets $L^*_v \subseteq L'_v$ and $k$ labels $f_c(v)$.
Given two labelling functions $f_v$ and $f_c$, the time complexity to compute the entries  $\Pi[L'_v,M,z]$
is $O(2^{k(2l+1)} k n)$.
By the property of color-coding~\cite{Alon:Yuster:Zwick:1995}, a function $f_v \in F_v$ and a function $f_v \in F_v $ 
can be computed in time $2^{O(l k)}poly(n)$ and $2^{O(k)}poly(n)$, respectively, hence
in time $2^{O(l k)}poly(n)$ . 
\qed
\end{proof}
}


\section{FPT Inapproximation}
\label{sec:FPTInapprox}

Since \textsc{MaxCDP} and \textsc{MaxCDDP} are hard to approximate in poly-time and do not admit fixed-parameterized algorithm for parameter number of paths, it is  worth to investigate approximation in FPT time, i.e. find approximate solution with additional time.
Unfortunately, in this section, we show that both \textsc{MaxCDP} and \textsc{MaxCDDP} do not admit an FPT cost $\rho$-approximation, for any function $\rho$ of the optimum, unless FPT=W[1]. 
We will show the result by giving a reduction from the \textsc{Threshold Set} problem.
Marx showed that the {\textsc{Threshold Set}} problem does not admit a fpt cost $\rho$-approximation, for any function $\rho$ of the optimum, unless FPT=W[1]~\cite{DBLP:journals/jcss/Marx13}.

First, we introduce the definition of the \textsc{Threshold Set} problem. 

\problemopt{\textsc{Threshold Set}}{
a set $U$ of elements, a collection $\S= \{S_1, \dots, S_q \}$ of subsets of $U$ and a positive integer weight $w(S_i)$ for each $S_i \in \S$,
with $1 \leq i \leq q$.}{a set $T \subseteq U$ of maximum cardinality such that $|T \cap S_i| \leq w(S_i)$ for every $S_i \in \S$.}

The cost of a solution of \textsc{Threshold Set} is denoted by $|T|$. 
Notice that this problem can be seen as a generalization of the \textsc{Independent Set} problem; 
indeed, for a graph $G=(V,E)$, we can define $U = V$, $\S = E$ and $w(S) = 1$ for every set $S \in \S$. 

We will reduce \textsc{Threshold Set} to {\textsc{MaxCDP}} in polynomial time such that there is a ``one-to-one'' 
correspondence between the solutions of the two problems, 
therefore the inapproximability result transfers to {\textsc{MaxCDP}}, and then to {\textsc{MaxCDDP}}. 
The reduction is inspired by the one in~\cite{DBLP:journals/algorithms/BonizzoniDP13}, that shows inapproximability in polynomial time and W[1]-hardness of \textsc{MaxCDP}.

First, we design the reduction for the {\textsc{MaxCDP}} problem. 
Notice that we assume that we are given an ordering over the sets in $\S$ (i.e. $S_i < S_j, i<j$).
Consider an instance $(U,\S,w)$ of \textsc{Threshold Set}, we define a corresponding instance $(G=(V,E,f_C), s, t)$ of {\textsc{MaxCDP}}. 
The set $V$ of vertices is defined as follows:
\[
V = \{s, t\} \cup \{ s_i | i \in [|U|]\} \cup \{S_i^j | i \in [|\S|], 1 \leq j \leq w(S_i) \}
\]

The set of colors $C$ is defined as follows: $C=\{ c_i: i \in U \}$.

Now, we define the set $E$ of edges.
\begin{itemize}

\item for all $i \in [|U|]$, define an edge $\{s,s_i  \}$ colored by $c_i$ and an edge $\{s_i, S_q^j\}$ colored by $c_i$, for all $1 \leq j \leq w(S_q)$, where $q$ is the smallest index of a set $S_q \in \S$ such that $i \in S_q$, 

\item from each $S_q^j$, define an edge $\{ S_q^j, S_{q'}^{j'}\}$ colored by $c_i$, for all $1 \leq j' \leq w(S_{q'})$, such that $i \in S_q$, $i \in S_{q'}$, $q' > q$ and, for each $q <l<q'$, it holds $i \notin S_l$,

\item from each $S_q^j$, define an edge $\{ S_q^j,t \}$ colored by $c_i$, where $i \in S_q$ and for each $q' > q$ with $S_{q'} \in \S$, $i \notin S_{q'}$.

\end{itemize}


See~\autoref{fig:innaprox} for an example.
Now, we prove the main properties of the reduction.

\begin{figure}[ht]
	\begin{center}
		\begin{tikzpicture}[scale=1,transform shape]

\node[vertex] (s) at (-0.5,0)	{$s$};

\node[vertex] (s1) at (1,-1.5) {$s_{1}$};
\node[vertex] (s2) at (1,1.5) {$s_{2}$};
\node[vertex] (s3) at (1,0.5) {$s_{3}$};
\node[vertex] (s4) at (1,-0.5) {$s_{4}$};

\node[vertex] (S11) at (3.5,-2) {$S_1^1$};
\node[vertex] (S12) at (3.5,2) {$S_1^2$};

\node[vertex] (S21) at (6,0) {$S_2^1$};

\node[vertex] (S31) at (8.5,-2) {$S_3^1$};
\node[vertex] (S32) at (8.5,2) {$S_3^2$};

\node[vertex] (t) at (11,0) {$t$};

\draw (s) edge[ultra thick] node[midway,fill=white] {1} (s1);
\draw (s) edge[ultra thick] node[midway,fill=white] {2} (s2);
\draw (s) edge node[midway,fill=white] {3} (s3);
\draw (s) edge node[midway,fill=white] {4} (s4);

\draw (s1) edge[ultra thick] node[midway,fill=white] {1} (S11) (S11) edge[ultra thick] node[midway,fill=white] {1} (S21);
\draw (s1) edge node[midway,fill=white] {1} (S12) (S12) edge node[midway,fill=white] {1} (S21);
\draw (S21) edge[ultra thick] node[midway,fill=white] {1} (t);

\draw (s2) edge node[midway,fill=white] {2} (S11) (S11) edge[] node[midway,fill=white] {2} (S31);
\draw (s2) edge [ultra thick] node[midway,fill=white] {2} (S12)  (S12) edge[ultra thick,bend left=30] node[midway,fill=white] {2} (S32);
\draw (S12) edge[bend left] node[midway,fill=white] {2} (S31);
\draw (S11) edge[bend left] node[midway,fill=white] {2} (S32);
\draw (S31) edge[bend right=55] node[midway,fill=white] {2} (t);
\draw (S32) edge[ultra thick,bend left=45] node[midway,fill=white] {2} (t);

\draw (s3) edge[bend right=15] node[midway,fill=white] {3} (S11) (S11) edge[bend right=15] node[midway,fill=white] {3} (S31);
\draw (s3) edge node[midway,fill=white] {3} (S12)  (S12) edge[bend left=15] node[midway,fill=white] {3} (S32);
\draw (S12) edge[bend left=50] node[midway,fill=white] {3} (S31);
\draw (S11) edge[bend right] node[midway,fill=white] {3} (S32);
\draw (S31) edge[bend right] node[midway,fill=white] {3} (t);
\draw (S32) edge[bend right] node[midway,fill=white] {3} (t);

\draw (s4) edge node[midway,fill=white] {4} (S21);
\draw (S21) edge node[midway,fill=white] {4} (S31);
\draw (S21) edge node[midway,fill=white] {4} (S32);
\draw (S31) edge node[midway,fill=white,bend right=55] {4} (t);
\draw (S32) edge node[midway,fill=white,bend left=35] {4} (t);
%

		\end{tikzpicture}
	\end{center}
\caption{Sample construction of an instance of \textsc{MaxCDP} from an instance of \textsc{Threshold Set} with $\S = \{\{1,2,3\}, \{1,4\}, \{2,3,4\} \}, w(S_1) = 2, w(S_2)=1, w(S_3)=2$. A solution for this instance of \textsc{Threshold Set} could be $T' = \{1,2\}$, and we drawn with dark edges the corresponding disjoint paths for \textsc{MaxCDP}.}\label{fig:innaprox}
\end{figure}
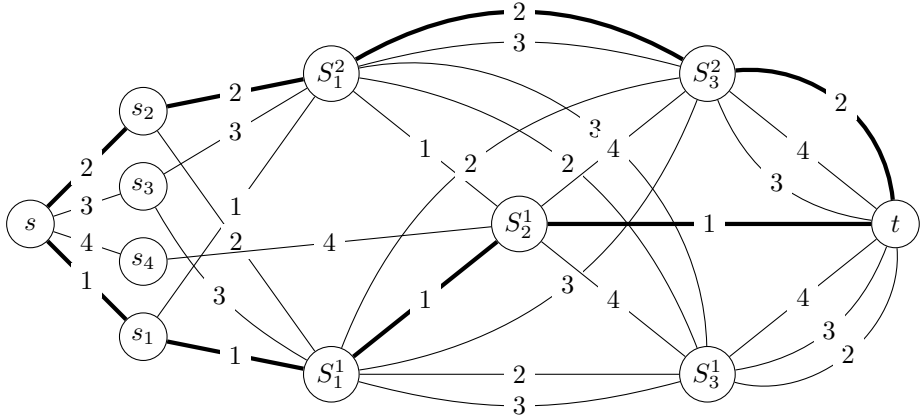

\begin{lemma}
\label{lemInapprox1}
Given an instance $(U,\S,w)$ of {\textsc{Threshold Set}}, let $(G=(V,E,f_C), s, t)$ be the corresponding instance of {\textsc{MaxCDP}}.
Then, given a solution $T'$ of {\textsc{Threshold Set}} on instance $(U,\S,w)$, we can compute in polynomial time 
a set of $|T'|$ disjoint uni-color paths in $(G=(V,E,f_C), s, t)$.
\end{lemma}
\begin{proof}
Consider a solution $T'$ of {\textsc{Threshold Set}} on instance $(U,\S,w)$, and define a set $P$ of $|T'|$ disjoint uni-color paths in $(G=(V,E,f_C), s, t)$ as follows. 
For each $i \in T'$, define a uni-color path $p$ colored by $c_i$ that starts in $s$, passes through $s_i$,  and for each $S_q \in \S$, if $i$ is the $j$-th element of $T'$ in $S_q$, $1 \leq j \leq w(S_q)$, passes through vertex $S_q^j$.
It follows that the path defined are disjoints, as at most one element can be the $j$-th element
of $T'$ in $S_q$ and $|T' \cap S_q| \leq w(S_q)$.
\qed
\end{proof}

\begin{lemma} 
\label{lemInapprox2}
\sloppy 
Given an instance $(U,\S,w)$ of {\textsc{Threshold Set}}, let $(G=(V,E,f_C), s, t)$ be the corresponding instance of {\textsc{MaxCDP}}.
Then, given a set of $q$ disjoint uni-color paths in $(G=(V,E,f_C), s, t)$, we can compute in polynomial time
a solution of size $q$ of {\textsc{Threshold Set}} on instance $(U,\S,w)$.
\end{lemma}
{
\begin{proof}
Consider a set $P$ of disjoint uni-color paths in $(G=(V,E,f_C), s, t)$. 
First, we claim that each path in $P$ has a distinct color.
Indeed, the paths in $P$ must be disjoint and, by construction, for each color $c_i$ each path must pass trough vertex $s_i$.

Now, starting from $P$, we define a solution $T'$ of of \textsc{Threshold Set} on instance $(U,\S,w)$.
For each path $p \in P$ colored by $c_i$, elements $u_i$ belongs to $T'$.
We show that $T'$ is a a solution of \textsc{Threshold Set} on instance $(U,\S,w)$.

Consider a set $S_i \in \S$, then there exists a most $w(S_i)$ elements in $T'$. 
Indeed, notice that, by construction, there exists at most $w(S_i)$ vertices $S_i^j$, hence by construction there exist at most $w(S_i)$ paths in $P$ that passes through vertices of $S_i^j$, hence at most $w(S_i)$ elements in $T'$ belong to $S_i$.
As a consequence $T'$ is a feasible solution of \textsc{Threshold Set} on instance $(U,S,w)$.
By construction, $|T'|=q$.
\qed
\end{proof}
}

\begin{theorem}
\label{teo:inapprox}
\textsc{MaxCDP} and \textsc{MaxCDDP} cannot be approximated in FPT-time within any function $\rho$ of the optimum, unless FPT=W[1].
\end{theorem}
{
\begin{proof}
The theorem holds for \textsc{MaxCDP} since \textsc{Threshold Set} cannot be approximated within any function $\rho$ of the optimum, 
unless FPT=W[1]~\cite{DBLP:journals/jcss/Marx13}, and from the properties of the polynomial time
reduction proved in Lemma~\ref{lemInapprox1} and Lemma~\ref{lemInapprox2}. 

For \textsc{MaxCDDP}, it holds from the fact that in the described reduction all the paths have a distinct color.
\qed
\end{proof}
}


%

\section{Conclusion}

In this paper, we continued the complexity analysis of \textsc{MaxCDP} 
and deepen the hardness analysis according to the structure of the input graph.
We also introduced a new variant, called \textsc{MaxCDDP}, asking for a solution with vertex and disjoint colors.

In the future, we would like to further deepen the analysis on the the structural complexity of \textsc{MaxCDP} and \textsc{MaxCDDP}. 
For example is \textsc{MaxCDP} in XP when the parameter if the size of the Feedback Vertex Set of the input graph? 
Is \textsc{MaxCDP} FPT when the parameter if the distance to disjoint paths of the input graph? 
Moreover, we have shown that \textsc{MaxCDDP} can be solved within a factor $\frac{1}{2}$
by a parameterized algorithm, where the parameter is the size of the vertex cover. A natural question
is whether \textsc{MaxCDDP} is fixed-parameter tractable when parameterized by the size of the vertex cover.

We would also like to improve the running time of our algorithms and to match them with some lower bounds 
under widely believed assumptions in order to have a fine-grained complexity analysis of these problems.


\bibliographystyle{abbrv}

\bibliography{biblio}


\end{document}